\documentclass[aps,prl,twocolumn]{revtex4-1}
\usepackage[latin9]{inputenc}
\usepackage{amsmath}
\usepackage{amssymb}
\usepackage{mathtools}
\usepackage{graphicx}
\usepackage{braket}
\usepackage{paralist}
\usepackage{color}
\usepackage{mathrsfs}
\usepackage{amsthm}
\usepackage{centernot}
\usepackage{stmaryrd}
\usepackage{xcolor}
\usepackage{soul}
\usepackage{thmtools} 
\usepackage{thm-restate}

\usepackage{url}

\usepackage{breakurl}
\usepackage[breaklinks]{hyperref}

\theoremstyle{definition}

\newtheorem*{proposition*}{Proposition}

\theoremstyle{plain}

\theoremstyle{definition}

\declaretheorem[name=Lemma]{lemma}
\declaretheorem[name=Definition]{definition2}

\newcommand{\Var}[0]{\textnormal{Var}}
\DeclareMathOperator*{\Tr}{Tr}

\DeclareMathOperator*{\Span}{Span}

\DeclareMathOperator{\MU}{MU}
\DeclareMathOperator{\image}{Im}

\DeclareRobustCommand{\openzero}{\leavevmode\hbox{0\kern-.55em0}}

\newcommand{\bs}{\boldsymbol}
\newcommand{\D}[1]{\mathcal D_{\mathbb B_{#1}}}
\newcommand{\maj}[2]{\mathbb B_{#1} \succ ^{\mathbb B_0} \mathbb B_{#2}}
\newcommand{\majj}[2]{{\mathbb {#1} \succ\!\!\succ ^{\mathbb B_0} \mathbb {#2}}}

\newcommand{\po}{``$\;\succ ^{ \mathbb B_0}$''}
\newcommand{\poo}{``$\;{\succ\!\!\succ ^{ \mathbb B_0}}$''}
\newcommand{\cre}[1]{c^{(\text{rel})}_{\mathbb B_{#1}}}
\newcommand{\cg}[1]{{c}_{\,\mathbb B_{#1}}}

\newcommand{\prlsection}[1]{{\em {#1}.---~}}

\makeatletter
\newcommand{\xMapsto}[2][]{\ext@arrow 0599{\Mapstofill@}{#1}{#2}}
\def\Mapstofill@{\arrowfill@{\Mapstochar\Relbar}\Relbar\Rightarrow}

\makeatother

\makeatletter

\begin{document}

\title{Quantifying the incompatibility of quantum measurements relative to a basis}

\date{\today}

\author{Georgios Styliaris}
\affiliation{Department of Physics and Astronomy, and Center for Quantum Information Science and Technology, University of Southern California, Los Angeles, California 90089-0484}

\author{Paolo Zanardi}

\affiliation{Department of Physics and Astronomy, and Center for Quantum Information
Science and Technology, University of Southern California, Los Angeles,
California 90089-0484}

\begin{abstract}

Motivated by quantum resource theories, we introduce a notion of incompatibility for quantum measurements relative to a reference basis. The notion arises by considering states diagonal in that basis and investigating whether probability distributions associated with different quantum measurements can be converted into one another by probabilistic post-processing. The induced preorder over quantum measurements is directly related to multivariate majorization and gives rise to families of monotones, i.e., scalar quantifiers that preserve the ordering. For the case of orthogonal measurement we establish a quantitative connection between incompatibility, quantum coherence and entropic uncertainty relations. We generalize the construction to arbitrary POVM measurements and report complete families of monotones.

\end{abstract}

\maketitle

\prlsection{Introduction}
One of the cornerstones of quantum theory is the concept of incompatibility between observables \cite{bohm2012quantum}. A pair of quantum observables is deemed incompatible if the corresponding self-adjoint operators fail to commute. Operationally, incompatibility implies that there exist pure quantum states for which it is impossible to simultaneously predict with certainty the measurement outcomes of two incompatible observables.
Finite-dimensional observables that share the same eigenbasis are fully compatible, while any pair of observables associated with bases that are mutually unbiased
are maximally incompatible: certain knowledge for the outcome of one assures complete randomness for the possible outcomes of the other.

Incompatibility is famously captured through uncertainty relations, that may involve variances \cite{PhysRev.34.163,Schroedinger,PhysRevLett.113.260401}, entropies \cite{PhysRevLett.50.631,PhysRevLett.60.1103,christandl2005uncertainty,PhysRevLett.108.210405,PhysRevA.89.042122,RevModPhys.89.015002,halpern2018reconciling} or other information-theoretic quantities \cite{PhysRevA.72.042110,PhysRevA.73.022324,PhysRevA.79.052106,gibilisco2007uncertainty,gibilisco2009quantum,PhysRevA.96.022132}. A quantitative description of incompatibility in quantum mechanics was persued recently, from the perspective of state discrimination and quantum steering \cite{PhysRevA.92.022115,PhysRevLett.115.230402,heinosaari2015incompatibility,Haapasalo_robustness_15,PhysRevA.93.052112,PhysRevA.98.012126,carmeli2018quantum,uola2018quantifying,skrzypczyk2019all}. In this approach, one of the central notions is that of a \textit{parent} measurement, i.e., one that can simulate the original one through probabilistic post-processing.

Quantum resource theories provide a framework to systematically characterize and quantify quantum properties (for example, entanglement). There, such a property is fully described by the conversion relations among states under a class of quantum processes that, suitably chosen, cannot enhance it \cite{chitambar_quantum_2018}. The transformation relations among quantum states can be mathematically described by a \textit{preorder}: if a state can be transformed into another under the distinguished class of processes, then it lies ``higher'' in the ordering \footnote{Mathematically, preorder (or quasiorder) is a binary relation that is reflexive and transitive}. In turn, the preorder induces a family of scalar functions, called \textit{monotones}, that cannot increase under the allowed state transitions and therefore jointly quantify the resourcefulness of states.

In this work, we introduce a notion of incompatibility of quantum measurements relative to a reference basis by means of a preorder. More specifically, considering states that are diagonal in the reference basis, we investigate whether the probability distributions associated to different measurements can be transformed into one another, by means of probabilistic post-processing. The aforementioned question of convertibility generates a preorder over quantum measurements which, in turn, gives rise to families of scalar functions that jointly quantify the introduced notion of incompatibility relative to a basis. We first consider the special case of orthogonal measurements in which the ordering provides a quantitative, as well as conceptual, connection between incompatibility, quantum coherence and entropic uncertainty relations. We then extend to include generalized measurements and we relate the resulting notion to parent measurements.


\prlsection{Preliminaries} Consider a non-degenerate observable $A$ over a finite dimensional Hilbert space $\mathcal H \cong \mathbb C ^d$ with spectral decomposition $A = \sum_{i=1}^d a_i P_i$ (we denote $P_i$ $ \coloneqq \ket{i}\!\bra{i}$). The role of the eigenvalues $a_i$ is to label the possible outcomes and, as long as they are distinct, this role is unimportant from the point of view of the measurement process, since the probability distribution $\bs p_{\mathbb B} (\rho)$ with components $[\bs p_{\mathbb B}(\rho)]_i \coloneqq \Tr \left( P_i \rho \right)$ (representing  a measurement of $A$ in state $\rho$) only depends on the set of projectors $\{ P_i\}_i$ \footnote{In this work, we will not distinguish among probability distributions that differ solely by permutations.}. We will henceforth use the term \textit{basis} (always meaning orthonormal) to refer to a set of rank-1 orthogonal projectors $\mathbb B = \{ P_i \}_{i=1}^d$, with $\sum_i P_i = I$ \footnote{Notice that such a definition does not distinguish between orthonormal sequences of kets that differ solely by reordering of elements or by phase factors, i.e., by a transformation $\ket{j} \mapsto e^{i \theta_j} \ket{\sigma (j)} $ ($\sigma \in \mathcal P_d$ is a permutation).}. A generalized measurement (POVM) is represented by a set of operators $\mathbb F = \{ F_i \}_{i}$ such that $F_i \ge 0$ and $\sum_i F_i = I$.
We associate with every basis $\mathbb B$ the real abelian algebra of observables $\mathcal A_{\mathbb B}$ generated by $\{P_i\}_i$. The set of bases over the Hilbert space is denoted by $\mathcal M(\mathcal H) $.

\prlsection{Preorder and monotones}
The idea of deriving families of scalar functions that quantify some feature (for instance, the degree of uniformity of a probability distribution) by invoking a preorder has its roots in the mathematical theory of majorization \cite{hardy1988inequalities}. Such a paradigm has been extensively employed in quantum information in the context of resource theories for quantifying features of quantum systems, such as entanglement \cite{horodecki2009quantum}, coherence \cite{PhysRevLett.113.140401} and out-of-equilibrium thermodynamics \cite{PhysRevLett.111.250404}.

In this approach, one distinguishes a class of quantum operations, deemed as ``easy'', motivated by some practical consideration.
For example, in the case of entanglement, the easy operations are local quantum operations between two parties together with classical communication (LOCC).
This set of maps induces a preorder "$\ge$" in the set of quantum states, defined by the allowed transitions under easy operations, namely $\rho \ge \sigma$ if and only if there exists an easy operation $\mathcal E$ such that $\sigma = \mathcal E (\rho)$.
The binary relation induced is a preorder since, by definition, the identity quantum channel is always an easy operation and also the composition of easy operations is again an easy operation. Moreover, $\rho \ge \sigma$ should intuitively correspond in our example to a statement like ``$\rho$ is more entangled than $\sigma$.'' This quantification is rigorously captured by the notion of \textit{monotones}, i.e., scalar functions $f$ over states, non-increasing under allowed state transitions ($\rho \ge \sigma \Longrightarrow f(\rho) \ge f(\sigma)$).
Families of monotones $\{ f_a \}_\alpha$ are said to form a \textit{complete set}, if they satisfy $f_\alpha(\rho) \ge f_\alpha(\sigma) \; \forall \alpha \; \Longleftrightarrow \rho \ge \sigma $.



\prlsection{A preorder over orthonormal bases} Our goal is to define a notion of incompatibility relative to a basis. Let us begin with the case of orthogonal measurements. Consider a basis $\mathbb B_0 = \{  P_i^{(0)} \}_{i}$ and a state $\rho_0 = \sum_i p_i P_i ^{(0)} \in \mathcal A_{\mathbb B_0}$ diagonal over it, described by the probability distribution $\bs p$. Given another basis $\mathbb B_1 = \{  P_i^{(1)} \}_{i}$, one can also associate with $\rho_0$ the probability distribution $\bs p_{\mathbb B_1}(\rho_0)$ corresponding to a measurement over $\mathbb B_1$. In fact, $\bs p_{\mathbb B_1}(\rho_0) =  X(\mathbb B_1 , \mathbb B_0) \bs p$, where $X(\mathbb B_1, \mathbb B_0)$ denotes the bistochastic matrix \footnote{Notice that the ordering of the projectors in a basis is arbitrary, hence the $X$ matrix is non-unique up to permutations.} with elements
\begin{align}
\left[X(\mathbb B_1, \mathbb B_0) \right]_{ij} \coloneqq \Tr \left( P_i^{(1)} P_j ^{(0)}  \right) \,\;.
\end{align}
Moreover, the probability distribution $\bs p_{\mathbb B_1} (\rho_0)$ is always ``more uniform'' than $\bs p$. This is precisely captured by the majorization statement $\bs p \succ \bs p _{\mathbb B_1} (\rho_0) $
that is true for any basis $\mathbb B_1$ and follows directly from the bistochasticity of $X$ \cite{marshall1979inequalities}.

Let us now introduce another measurement, over a basis $\mathbb B_2$, such that there exists some bistochastic matrix $M$ with
\begin{align} \label{Eq:prob_post_process}
X(\mathbb B_2 , \mathbb B_0) = M X(\mathbb B_1 , \mathbb B_0) \,\;.
\end{align} 
This relation has a rather strong implication: for all states $\rho_0$ diagonal in $\mathbb B_0$, the distribution $\bs p_{\mathbb B_2} (\rho_0) $ can be obtained from $\bs p_{\mathbb B_1} (\rho_0) $ through ``uniforming'' classical post-processing, represented by some bistochastic $M$ which is independent of the state.

Motivated by the above, if Eq.~\eqref{Eq:prob_post_process} holds, we declare that ``an orthogonal measurement over $\mathbb B_1$ is more compatible than over $\mathbb B_2$, relative to states diagonal in $\mathbb B_0$''. We introduce the following notation.
\begin{definition2}
\label{th:maindefinitionnew}
We denote $\mathbb B_1 \succ ^{ \mathbb B_0}  \mathbb B_2$ if and only if there exists a bistochastic matrix $M$ such that $X(\mathbb B_2 , \mathbb B_0) = M X(\mathbb B_1 , \mathbb B_0)$.
\end{definition2}
The definition has the following immediate consequences.
\begin{inparaenum}[(i)]
\item The binary relation \po over $\mathcal M (\mathcal H)$ is a preorder, i.e., $\maj{}{}$ $\forall \,\mathbb B$ (reflexivity) and  $\maj{1}{2}$, $\maj{2}{3}$ $\Longrightarrow$ $\maj{1}{3}$ (transitivity).
\item $\maj{0}{}$ for all bases $\mathbb B$ (``measurement over $\mathbb B_0$ is more compatible than over any other basis'')
\item  $\mathbb B \succ ^{ \mathbb B_0} \mathbb B_{\MU}$ for all bases $\mathbb B$, where $\mathbb B_{\MU}$ is any basis mutually unbiased to $\mathbb B_0$ (``measurement over any basis is more compatible than over a mutually unbiased one'').
\end{inparaenum}

The preorder \po is not in general a partial order, i.e., $\maj{1}{2}$ and $\maj{2}{1}$ do not necessarily imply $\mathbb B_1 = \mathbb B_2$ . For example, any $\mathbb B_1$ and $\mathbb B_2$ that are unbiased relative to $\mathbb B_0$ satisfy the aforementioned relations but can be taken to be distinct.

The ordering \eqref{Eq:prob_post_process} over matrices has been studied in the context of multivariate majorization, called \textit{matrix majorization} \cite{marshall1979inequalities}. There, $A \succ C$ for matrices $A$ and $C$ if there exists a bistochastic $B$ such that $C = B A$. We now connect the aforementioned preorder with quantum measurements.


\prlsection{\po from non-selective measurements} Def.~\ref{th:maindefinitionnew}
can be operationally understood in terms of classical post-processing of probability distributions. Here we show that the ordering \po also admits a quantum operational interpretation in terms of emulation of a non-selective measurement via additional such measurements.

Any basis $\mathbb B$ gives rise to a corresponding \textit{dephasing} or \textit{measurement} quantum map
\begin{align}
	\mathcal D_{\mathbb  B} (X) \coloneqq \sum_i P_i X P_i \,\;. \label{dephasing_definition}
\end{align}
The latter can be though of as a non-selective orthogonal measurement of any non-degenerate observable belonging in $\mathcal A_{\mathbb B}$, while a composition $\mathcal D_{\mathbb B_n} \dots \mathcal D_{\mathbb B_1}$ represents the quantum operation associated with $n$ such successive measurements \footnote{In fact, the basis $\mathbb B$ corresponding to a dephasing map $\mathcal D_{\mathbb B}$ is unique, i.e.,  the mapping $\mathbb B \mapsto \mathcal D_{\mathbb B}$ is injective  \cite{PhysRevA.97.032304}, and similarly for $\mathbb B \mapsto \mathcal A_{\mathbb  B}$ \cite{zanardi_18}.}.

We are now ready to state the result.
The ordering $\maj{1}{2}$ holds if and only if, for any initial state diagonal in $\mathbb B_0$, the output of a non-selective $\mathbb B_2$ measurement can be emulated by a non-selective $\mathbb B_1$ measurement, followed possibly by an additional sequence of measurements and a unitary rotation. More specifically:

\begin{restatable}{prop}{preorderdephasing}
\label{th:maindefinition}
$\maj{1}{2}$ if and only if there exist a unitary superoperator $\,\mathcal U$ and a (possibly trivial) sequence of measurements $\{ \mathcal D_{\mathbb B'_\alpha} \}_\alpha$ such that
\begin{align}
\mathcal D_{\mathbb B_2} \mathcal D_{\mathbb B_0} = \mathcal U \left[ \prod _\alpha \mathcal D_{\mathbb B'_\alpha} \right]  \mathcal D_{\mathbb B_1} \mathcal D_{\mathbb B_0} \,\;.  \label{eq_preorder_dephasing}
\end{align}
\end{restatable}
All proofs can be found in the Appendix.

The auxiliary sequence of measurements needed might be, in fact, infinite. Eq.~\eqref{eq_preorder_dephasing} should be understood as ``$\left\| \mathcal D_{\mathbb B_2} \mathcal D_{\mathbb B_0} - \mathcal U \left[ \prod _\alpha \mathcal D_{\mathbb B'_\alpha} \right]  \mathcal D_{\mathbb B_1} \mathcal D_{\mathbb B_0} \right\|$ can be made arbitrarily small'', i.e., the state transformation of the RHS can approximate arbitrarily well the one of the LHS.


We now analyze the $d=2$ case, by invoking Prop.~\ref{th:maindefinition} together the usual Bloch ball representation of quantum states $\rho = \frac{1}{2}\left( I +\bs v \cdot \bs \sigma \right)$, where different bases are in one to one correspondence with lines passing from the center. In this representation, the action of $\D{1}$ on a state $\rho$ coincides with projecting $\bs v$ onto the $\mathbb B_1$ line while the action of $\mathcal U$ is translated into an $SO(3)$ rotation. Clearly, Eq.~\eqref{eq_preorder_dephasing} can be satisfied (in fact, by means of a single $\mathcal D_{\mathbb B'_1}$) if and only if $\theta_{1} \le \theta_{2}$; here $\theta_i$ is the (acute) angle between the lines corresponding to $\mathbb B_0$ and $\mathbb B_i$. In particular, for $d=2$ the ordering \po is a total preorder, but not for $d > 2$.

\prlsection{Measures of relative (in)compatibility} A preorder gives rise to a distinguished class of scalar functions, i.e., monotones. We adopt the following definition.

\begin{definition2} A function $f_{\mathbb B_0} :  \mathcal M(\mathcal H) \to \mathbb  R_0^+$ is \textit{measure of compatibility (incompatibility)} relative to $\mathbb  B_0$ if it convex (concave) with respect to the preorder \po, i.e., $\maj{1}{2} \Longrightarrow f_{\mathbb B_0}(\mathbb B_1) \ge f_{\mathbb B_0}(\mathbb B_2)$ $\big(\,\maj{1}{2} \Longrightarrow  f_{\mathbb B_0}(\mathbb B_1) \le f_{\mathbb B_0}(\mathbb B_2)\,\big)$.
Moreover, if $f_{\mathbb B_0}(\mathbb B_1) = f_{\mathbb B_1}(\mathbb B_0)$, we call it a \textit{symmetric} measure of relative compatibility (incompatibility).
\end{definition2}


The following Proposition gives a construction for measures of relative compatibility arising from convex functions. It is a direct consequence of a result from \cite{karlin1983comparison}, derived in the context of matrix majorization.

\begin{restatable}{prop}{convexmonotones}
\label{th:convexmonotones}
Let $\phi : \mathbb R ^d \to \mathbb R$ be a continuous convex function.
Then,
\begin{align}
	f^{\phi}_{\mathbb B_0}(\mathbb B_1) \coloneqq \sum_{i} \phi(X_i^R(\mathbb B_1,\mathbb B_0)) \label{eq:measures_convex}
\end{align}
is a measure of relative compatibility; here, $X_i^R$ stand for the row vectors of the matrix $X_{ij}$.
\end{restatable}

An analogous claims hold for the incompatibility case in terms of concave functions.

In fact, the family $\{f^{\phi}_{\mathbb B_0}(\mathbb B_1) \}_\phi$ for all continuous convex $\phi$ is known to be a complete family of monotones for matrix majorization \cite{karlin1983comparison}, i.e., joint monotonicity $f^{\phi}_{\mathbb B_0}(\mathbb B_1) \ge f^{\phi}_{\mathbb B_0}(\mathbb B_2)$ for all such functions is enough to imply $\maj{1}{2}$. In that sense, the existence of a probabilistic uniforming process $M$ such that Eq.~\eqref{Eq:prob_post_process} holds is fully captured by this family of functions.


\prlsection{Incompatibility and coherence} Quantum coherence refers to the property of quantum systems to exist in a linear superposition of different physical states. It is a notion defined with respect to some preferred, physically relevant basis, which we will denote as $\mathbb B_0$.
A state $\rho$ is said to be coherent if there exist non-vanishing off-diagonal elements when $\rho$ is expressed as a matrix in $\mathbb B_0$. Recently, coherence was formulated as a resource theory \cite{RevModPhys.89.041003}.
One of the central measures in the theory is \textit{relative entropy of coherence}, $ \cre{0}(\rho) \coloneqq S(\rho \, \|  \,\D{0} \rho ) $ that admits several operational interpretations in terms of conversion rates \cite{PhysRevLett.116.120404,PhysRevLett.120.070403}. Later, we will also invoke the 2-coherence $c_{\mathbb B_0}^{(2)} \coloneqq \sum_{i\ne j} \left|\rho_{ij}\right|^2$ \footnote{The 2-coherence, however, fails to satisfy the monotonicity property under some sets of incoherent operations.}.


The ordering \po has rather strong implications in terms of quantum coherence, both in terms of state conversion under Incoherent Operations (i.e., the easy operation in the resource theory \cite{PhysRevLett.113.140401} of coherence) as well as coherence monotones. We define the action of a unitary superoperator over a basis as $\mathcal U (\mathbb B) \coloneqq  \{ \mathcal U (P_i)  \}_i$.

\begin{restatable}{prop}{relcoherence}
\label{th:relcoherence}
Let $\maj{1}{2}$.
\begin{enumerate}[(i)]
\item Consider a pair of unitary quantum maps $\mathcal U, \mathcal V$ such that $\mathcal U (\mathbb B_1) = \mathbb B_0$ and $\mathcal V (\mathbb B_2) = \mathbb B_0$ and a pure state $P_j \in \mathbb B_0$. Then, $\mathcal V (P_j)$ can be transformed to $\mathcal U (P_j)$ via incoherent operations over $\mathbb B_0$. Consequently, all coherence measures over such states are non-increasing.

\item $\cg{1} (\rho_0) \le \cg{2} (\rho_0)$ for all $\rho_0$ diagonal in $\mathbb B_0$, where $\cg{}$ denotes either the relative entropy of coherence or the 2-coherence over $\mathbb B$.

\end{enumerate}
\end{restatable}


In addition to the interpretation of Prop.~\ref{th:relcoherence} in the framework of coherence, one can also infer from (ii) above that a $\D{1}$ measurement disturbs less $\rho_0$ compared to a $\D{2}$ measurement, if $\maj{1}{2}$, as it is precisely captured by statistical meaning of the relative entropy \cite{wilde2013quantum}.

In the light of the interpretation of $\cre{}$ as distillable coherence \cite{PhysRevLett.116.120404}, (ii) above demonstrates a quantitative trade-off between compatibility and coherence. Moreover, any coherence average $C_{\,\mathbb B_0}(\mathbb B) \coloneqq \int d\mu (\rho_0) \cg{} (\rho_0)  $ is a measure of incompatibility of $\mathbb B$ relative to $\mathbb B_0$. In fact, these averages over the uniform distribution have been performed, verifying explicitly that $C_{\,\mathbb B_0}(\mathbb B)  = f^{\phi}_{\mathbb B_0}(\mathbb B_1)$ is of the form indicated in the (concave analogue of) Prop.~\ref{th:convexmonotones}. Indeed, $\phi$ coinsides with the subentropy \cite{jozsa1994lower} for the case of the relative entropy of coherence, while  $\phi (p_1,\dots,p_d) \propto \sum_i ( \frac{1}{d} - p_i^2 )$ for the 2-coherence \cite{zhang_coherence_2018}.

Finally, we note that in \cite{durt2010mutually}, the authors considered a geometrically motivated measure of ``mutual unbiasedness'' between pairs of orthonormal bases. Their measure is proportional to the 2-coherence average above, hence is also a symmetric measure of relative incompatibility.



\prlsection{Incompatibility and uncertainty} We now consider implication of the preorder \po in terms of uncertainty and fluctuations.

By its definition, the ordering $\maj{1}{2}$ assures that the distribution $\bs p _{\mathbb B_2} (\rho_0) $ is ``more uniform'' than $\bs p _{\mathbb B_1} (\rho_0)$, for any state $\rho_0$ diagonal in $\mathbb B_0$. An immediate consequence is that all Schur-concave functions, which for instance include $\alpha$-R\'enyi entropies for ($\alpha = 1$ corresponds to the usual Shannon entropy), satisfy $S_\alpha(\bs p _{\mathbb B_1}(\rho_0)) \le S_\alpha( \bs p _{\mathbb B_2}(\rho_0))$ \footnote{The converse statement in terms of majorization, however, does not hold, i.e., $\bs p _{\mathbb B_1} (\rho_0) \succ \bs p _{\mathbb B_2} (\rho_0) $ for all diagonal $\rho_0$ is not enough to assure $\maj{1}{2}$. A specific counterexample was constructed in the context of multivariate majorization by Horn in \cite{horn1954doubly}.}.

Quantum fluctuations over different bases can be quantified via entropic uncertainty relations \cite{RevModPhys.89.015002}. There, one tries to impose bounds over entropic quantities, such as $S_\alpha(\bs p_{\mathbb B_1}(\rho_0)) + S_\alpha(\bs p_{\mathbb B_2}(\rho_0)) \ge r_{\mathbb B_0}(\mathbb B_2,\mathbb B_1)$ ($\alpha = \beta = 1$ corresponds to the usual Shannon entropy), as a function of the bases. The most well-known inequality is due to Maassen and Uffink \cite{PhysRevLett.60.1103} and states that a ($\mathbb B_0$ independent) choice for the above bound is $r^{(\text{MU})}(\mathbb B_2,\mathbb B_1) \coloneqq  -\log (\max_{i,j}X_{ij}(\mathbb B_2,\mathbb B_1)) $ for any $\alpha,\beta\ge 1/2$ with $1/\alpha + 1/\beta = 2$ . The bound has recently been improved by Coles \textit{et al.} \cite{PhysRevLett.108.210405} for the case of Shannon/von Neumann entropy, as $S(\bs p_{\mathbb B_1}(\rho_0)) + S(\bs p_{\mathbb B_2}(\rho_0))  \ge S(\rho_0) + r^{(\text{MU})}(\mathbb B_2,\mathbb B_1)$.

Let us also consider the quantity
\begin{align}
Q_{\mathbb B_ 0} (\mathbb B_1) &\coloneqq \sup_{A \in \mathcal A_{\mathbb B_1}, \left\|A  \right\|_2 = 1} \; \max_{i = 1,\dots,d} \Var_i (A) \,\;,\\
\text{where }\Var_i (A) &\coloneqq \Tr \left( P_i^{(0)} A^2 \right) - \left[ \Tr \left( P_i^{(0)} A \right) \right]^2 \nonumber \,\;,
\end{align}
that captures the strength of the fluctuations of a pure state diagonal in $\mathbb B_0$ over a $\mathbb B_1$ measurement. In the Appendix we derive the upper bound
\begin{align}
Q_{\mathbb B_0} (\mathbb B_1)  \le 1 - \lambda_{\min} \left( X(\mathbb B_1,\mathbb B_0)X^T(\mathbb B_1,\mathbb B_0) \right) \coloneqq q(\mathbb B_1 , \mathbb B_0)
\end{align}
($\lambda_{\min}(X)$ stands for the minimum eigenvalue of $X$). 
The bound is symmetric and satisfies $q(\mathbb B_1 , \mathbb B_0) = 0$ if and only if $\mathbb B_1 = \mathbb B_0$, hence it vanishes if and only if  $Q_{\mathbb B_0}(\mathbb B_1)$ vanishes.

In words, $r^{(\text{MU})}$ and $q$ provide bounds on uncertainty and fluctuations that arise due to the incompatibility between the bases of measurement (for $r^{(\text{MU})}$) or state preparation and measurement (for $q$), and can be thought of as playing a role analogous to the commutator term in the usual uncertainty relations for observables. As such, they both turn out to be (symmetric) measures of relative incompatibility, monotonic relative to the ordering \po.

\begin{restatable}{prop}{bounds}
\label{th:bounds}
Let $\maj{1}{2}$. Then, $q(\mathbb B_1 , \mathbb B_0) \le q(\mathbb B_2 , \mathbb B_0)$ and $r^{(\text{MU})}(\mathbb B_1,\mathbb B_0) \le r^{(\text{MU})}(\mathbb B_2,\mathbb B_0)$.
\end{restatable}

\prlsection{Generalized measurements} The ordering \po can be directly extended to include generalized measurements described by POVMs. Consider a state $\rho_0 = \sum_i p_i P_i ^{(0)} \in \mathcal A_{\mathbb B_0}$ and a measurement $\mathbb F = \{ F_i \}_{i}$. The probability distribution of possible outcomes is $\bs p_{\mathbb F}(\rho_0) = X(\mathbb F, \mathbb B_0) \bs p$, where now $\left[X(\mathbb F, \mathbb B_0) \right]_{ij} \coloneqq \Tr ( F_i P_j ^{(0)}  )$ is just column stochastic \footnote{The POVMs are allowed have arbitrary number of elements. If this number is different, it is understood that the POVM with the least number of elements is padded with zeros until the cardinality of the sets becomes equal, so that the (rectangular) matrices $X(\mathbb F,\mathbb B_0)$ and $X(\mathbb G,\mathbb B_0)$ have equal dimensions.}. The analogous ordering over POVMs $\mathbb F$ and $\mathbb G$ relative to a basis $\mathbb B_0$ can be defined as $\mathbb F \succ^{\mathbb B_0} \mathbb G$ if and only if there exists a bistochastic $M$ such that $X(\mathbb G , \mathbb B_0) = M X(\mathbb F , \mathbb B_0)$.
In fact, the family $\{f^{\phi}_{\mathbb B_0}(\mathbb F) \coloneqq \sum_{i=1}^d \phi(X_i^R(\mathbb F,\mathbb B_0))  \}_\phi$ for all continuous convex $\phi$ still forms a complete family of monotones for the ordering \po, now considered over POVMs.

However, in contrast with the orthogonal measurement case, now it does not hold that $\bs p_{\mathbb B_0} (\rho_0) \succ \bs p _{\mathbb F} (\rho_0) $ for all $\mathbb F$, namely generalized measurements can ``purify'' the initial probability distribution \footnote{As, for example, with $\mathbb F = \{ I,0,\dots,0\}$.}. For this reason, we consider as the appropriate meaningful generalization of ``incompatibility relative to a basis'' to POVMs the less restraining ordering that occurs by relaxing the constraint of bistochasticity  on the matrix $M$, and instead requiring only column stochasticity. In this case, if $\mathbb F$ lies ``higher'' in the ordering than $\mathbb G$, then $\bs p _{\mathbb G} (\rho_0)$ can be obtained by probabilistic post-processing (not necessarily a uniforming one) from $\bs p _{\mathbb F} (\rho_0)$, independently of $\rho_0 \in \mathcal A_{\mathbb B_0}$.

\begin{definition2}
\label{th:maindefinitiogeneral}
We denote $\majj{F}{G}$ if and only if there exists a stochastic matrix $M$ such that $X(\mathbb G , \mathbb B_0) = M X(\mathbb F , \mathbb B_0)$.
\end{definition2}

The ordering is a preorder and clearly $\mathbb F  \succ ^{\mathbb B_0} \mathbb G  \Longrightarrow \majj{F}{G}$. As such, the corresponding monotones for \poo are related to Eq.~\eqref{eq:measures_convex}. The following is a direct implication of a result by Alberti \textit{et al.}~\cite{alberti1982wachsende} (see also \cite{alberti2008order}).

\begin{restatable}{prop}{convexmonotones2}
\label{th:convexmonotones2}
Let $\psi : \mathbb R ^d \to \mathbb R$ be a function that is simultaneously convex and homogeneous in all its arguments. Then,
\begin{align}
	g^{\psi}_{\mathbb B_0}(\mathbb F) \coloneqq \sum_{i} \psi(X_i^R(\mathbb F,\mathbb B_0))  \label{eq:measures_convex2}
\end{align}
is a monotone over \poo, i.e., ${\majj{F}{G}} \Longrightarrow g^{\psi}_{\mathbb B_0}(\mathbb F) \ge g^{\psi}_{\mathbb B_0}(\mathbb G)$; here, $X_i^R$ stand for the row vectors of the matrix $X_{ij}$. Moreover, the family $\{ g^{\psi}_{\mathbb B_0}(\mathbb F)  \}_\psi$ forms a complete set of monotones for \poo.
\end{restatable}

\prlsection{Basis-independent incompatibility} Finally, we connect the orderings describing measurement incompatibility relative to a basis with the notion of a \textit{parent measurement} \cite{heinosaari2016invitation,skrzypczyk2019all}. In this context, $\mathbb F$ is called a parent of $\mathbb G$ if there exists a stochastic $M$ such that $G_i = \sum_j M_{ij} F_j$ $\forall i$, while a family of measurements are jointly measurable if they admit a common parent.

\begin{restatable}{prop}{allbases}
\label{th:allbases}
$\mathbb F$ is a parent of $\mathbb G$ if and only if $\majj{F}{G}$ for all $\mathbb B_0 \in \mathcal M (\mathcal H)$ and the post-processing matrix $M$ can be chosen to be the same for all $\mathbb B_0$.
\end{restatable}

\prlsection{Conclusions} Quantum resource theories seem to suggest that an appropriate quantification of quantum properties, even conceptually simple ones such as the ``uniformity'' of a state \cite{Gour_uniformity}, cannot be achieved by means of a single scalar quantifier. Instead, only an infinite set of functions is able to capture such properties in their wholeness, as they naturally result out of preorders. In this work, we defined an operationally motivated preorders over quantum measurements that capture a notion of incompatibility relative to a basis. Our approach uncovers a quantitative, as well as conceptual, connection between incompatibility, uncertainty relations and quantum coherence unified under the prism of multivariate majorization.

\acknowledgements

\prlsection{Acknowledgements} G.S. is thankful to N.A. Rodr\'{i}guez-Briones for helpful discussions and acknowledges financial support from a University of Southern California ``Myronis'' fellowship. P.Z. acknowledges partial support from the NSF award PHY-1819189.

\Urlmuskip=0mu plus 1mu\relax
\bibliography{refs}

\appendix
\widetext

\section{Appendix: Proofs}

\begin{proposition*}
\begin{inparaenum}[(i)]
\item The binary relation \po over $\mathcal M (\mathcal H)$ is a preorder, i.e., $\maj{}{}$ $\forall \,\mathbb B$ (reflexivity) and  $\maj{1}{2}$, $\maj{2}{3}$ $\Longrightarrow$ $\maj{1}{3}$ (transitivity).
\item $\maj{0}{}$ for all bases $\mathbb B$ (``measurement over $\mathbb B_0$ is more compatible than over any other basis'')
\item  $\mathbb B \succ ^{ \mathbb B_0} \mathbb B_{\MU}$ for all bases $\mathbb B$, where $\mathbb B_{\MU}$ is any basis mutually unbiased to $\mathbb B_0$ (``measurement over any basis is more compatible than over a mutually unbiased one'').
\end{inparaenum}
\end{proposition*}
\begin{proof}
\textbf{(i)} Reflexivity follows since $I$ is bistochastic and transitivity from the fact that a product of bistochastic matrices is also bistochastic.
\\ \\
\textbf{(ii)} Since $X(\mathbb B_0,\mathbb B_0) = I$, follows by setting $M = X(\mathbb B,\mathbb B_0)$
\\ \\
\textbf{(iii)} By definition, $\left[ X(\mathbb B_{\MU},\mathbb B_0) \right]_{ij} = 1/d$, hence follows by setting $M _{ij} = 1/d$.
\end{proof}

Let us now establish a helpful Lemma. We remind the reader that a bistochastic matrix $A_{ij}$ is \textit{unistochastic} \cite{bengtsson2017geometry} if there exists a unitary matrix $U_{ij}$ such that $A_{ij} = | U_{ij} |^2$.

\begin{lemma}
Every bistochastic matrix can be approximated arbitrarily well by a product of unistochastic matrices.
\end{lemma}
\begin{proof}

Assume $M$ is a bistochastic matrix such that $M_{ij}> 0 $ for all $i,j$. Then, $M$ can be expanded into a finite product of T-transform \cite{bhatia2013matrix}, which are unistochastic matrices. This is because T-transforms act non-trivially only on a 2-dimensional subspace and all bistochastic matrices in $d=2$ are unistochastic.

The set of bistochastic matrices forms a convex polytope and hence in any $\epsilon$-neighbourhood (as defined, e.g., by the $l_1$ norm) of a matrix $M$ that fails the element-wise positivity condition, there exists some $M'$ that fulfills it.
\end{proof}

\preorderdephasing*

\begin{proof}


Eq.~\eqref{eq_preorder_dephasing} holds if and only if the action of the LHS and the RHS on any $P_i^{(0)}$ coincide. This is because $\D{0}$ is a projector and hence the action is non-trivial only over the image $\image (\D{0}) = \Span \left\{ P_i^{(0)} \right\}_i$. We have,
\begin{subequations} \label{lhsrhs}
\begin{align}
	&\text{LHS: } \quad \D{2} \D{0}  P_i^{(0)} = \sum_{j} X_{ji}(\mathbb B_2,\mathbb B_0) P_j^{(2)}  \label{lhs}\\
	&\text{RHS: }\quad \mathcal U \left[ \prod _\alpha \mathcal D_{\mathbb B'_\alpha} \right]  \mathcal  D_{\mathbb B_1} \mathcal D_{\mathbb B_0} P_i^{(0)}  = \sum_{\{ j_\alpha\}}\left[ \prod_{\alpha = 1} ^{\alpha_{max}-1}  X_{j_{\alpha+1} \, j_\alpha}(\mathbb B'_{\alpha+1},\mathbb B'_{ \alpha}) \right] X_{j_1i}(\mathbb B_1,\mathbb B_0) \, \mathcal U \left(P_{j_{\alpha_{\max}}}^{(\alpha_{\max})} \right) \,\;. \label{rhs}
\end{align}
\end{subequations}
Notice, in addition, that an appropriate $\mathcal U$ for the two expressions to be equal should satisfy $\mathbb B_2 = \mathcal U (\mathbb B'_{\alpha_{\max}})$.
\\ \\
Let us first show sufficiency. If Eq.~\eqref{eq_preorder_dephasing} holds, then the expressions \eqref{lhsrhs} are equal and therefore one can directly see that Eq.~\eqref{Eq:prob_post_process} also holds for bistochastic $M = \prod _\alpha A^{(\alpha)}$, where $A^{(\alpha)} = X(\mathbb B'_{\alpha+1},\mathbb B'_{ \alpha})$.
\\ \\
We now prove necessity.
Assume $\maj{1}{2}$, hence there exists a bistochastic $M$ such that Eq.~\eqref{Eq:prob_post_process} holds. Now, with use of the Lemma, we decompose $M = \prod_\alpha A^{(\alpha)}$ into a product of unistochastic matrices. For all $A^{(\alpha)}$ there exist, by definition, unitary operators $\mathcal U^{(\alpha)}$ such that $A^{(\alpha)}_{ij} = \Tr \left( P_j^{(0)} \mathcal U^{(\alpha)} (P_i^{(0)})  \right)$ for all $i,j$, or equivalently, $A^{(\alpha)} = X(\mathcal U^{(\alpha)} (\mathbb B_0),\mathbb B_0)$. Now we show that Eq.~\eqref{eq_preorder_dephasing} also holds for a sequence of dephasing superoperators $\{ \mathcal D_{\mathbb B'_\alpha} \}_{\alpha=1}^{\alpha_{\max}}$ over the bases
\begin{align}
\mathbb B'_1 &= \mathcal W (\mathbb B_0) \\
\mathbb B'_\alpha & = \mathcal W \, \mathcal U^{(1)}  \mathcal U^{(2)} \dots \mathcal U^{(\alpha)} (\mathbb B_0)  \quad \text{for all } 1 \le \alpha \le \alpha_{\max} \,\;,
\end{align}
where $\mathcal W(\mathbb B_0) = \mathbb B_1$. To see that, fist notice that for any unitary superoperator $\mathcal V$ it holds that $X(\mathbb B_\alpha,\mathbb B_{\beta}) = X(\mathcal V (\mathbb B_\alpha), \mathcal V(\mathbb B_{\beta})) $. As a result, we can write
\begin{align*}
X(\mathbb B_2,\mathbb B_0) &= \left[ \prod_\alpha  A^{(\alpha)} \right] X(\mathbb B_1,\mathbb B_0) =\left[ \prod_\alpha X(\mathcal U^{(\alpha)} (\mathbb B_0),\mathbb B_0)\right]  X(\mathbb B_1,\mathbb B_0)   \\
& = \, \dotso \, X(\mathcal W  \mathcal U^{(1)} \mathcal U^{(2)} (\mathbb B_0),\mathcal \mathcal W  \mathcal U^{(1)} (\mathbb B_0)) \, X(\mathcal W  \mathcal U^{(1)}(\mathbb B_0),\mathcal W (\mathbb B_0)) \,X(\mathcal W (\mathbb B_0),\mathbb B_0) \\
&= X(\mathbb B'_{\alpha_{\max}},\mathbb B'_{\alpha_{\max}-1}) \, \dotso \, X(\mathbb B'_2,\mathbb B'_1)  \, X(\mathbb B'_1,\mathbb B_1) \, X(\mathbb B_1,\mathbb B_0) \,\;.
\end{align*}
Choosing $\mathcal U$ such that $\mathbb B_2 = \mathcal U (\mathbb B'_{\alpha_{\max}})$, the above equation implies that the expressions \eqref{lhsrhs} are equal and hence Eq.~\eqref{eq_preorder_dephasing} also holds for the described sequence of dephasing superoperators.

\end{proof}



\convexmonotones*

\begin{proof}
Let $\maj{1}{2}$. Then, there exists a bistochastic matrix $M$ such that $X(\mathbb B_2,\mathbb B_0) = M X(\mathbb B_1,\mathbb B_0)$. For any continuous convex function $\phi : \mathbb R ^d \to \mathbb R$,
\begin{align*}
	&\sum_{i=1}^d \phi \left( X^R_i(\mathbb B_2,\mathbb B_0) \right) = \sum_i \phi \left( \sum_k M_{ik} X^R_k(\mathbb B_1,\mathbb B_0) \right) \\
	& \le \sum_{i,k} M_{ik} \phi \left(  X^R_k(\mathbb B_1,\mathbb B_0) \right) =  \sum_i \phi \left( X^R_i(\mathbb B_1,\mathbb B_0) \right) \,\;.
\end{align*}
%

\end{proof}

\relcoherence*

\begin{proof}
\textbf{(i)} In \cite{PhysRevA.91.052120} (see also \cite{PhysRevLett.116.120404}) it was shown that $\ket{\psi}\!\bra{\psi}$ can be transformed to $\ket{\phi}\!\bra{\phi}$ via Incoherent Operations (in fact, Strictly Incoherent Operations) with respect to $\mathbb B_0$ if $\mathcal D_{\mathbb B_0} (\ket{\phi}\!\bra{\phi}) \succ \mathcal D_{\mathbb B_0} (\ket{\psi}\!\bra{\psi}) $.

From the assumption $\maj{1}{2}$, we have that $X^C_{j} (\mathbb B_1,\mathbb B_0) \succ X^C_{j} (\mathbb B_2,\mathbb B_0)$ $\forall j$, where $X^C_{j}$ denotes the $j$th column vector of $X$. We can rewrite
\begin{align*}
[X^C_{j} (\mathbb B_1,\mathbb B_0)]_i  &= \Tr ( P_i^{(1)} P_j^{(0)}  ) = \Tr ( \mathcal U^\dagger (P_i^{(0)} ) P_j^{(0)}  ) \\
& = \Tr (  P_i^{(0)} \mathcal U  (P_j^{(0)} ) )
\end{align*}
and similarly
\begin{align*}
[X^C_{j} (\mathbb B_2,\mathbb B_0)]_i =  \Tr (  P_i^{(0)} \mathcal V  (P_j^{(0)} ) ) \,\;.
\end{align*}
Now, we can write the relation $X^C_{j} (\mathbb B_1,\mathbb B_0) \succ X^C_{j} (\mathbb B_2,\mathbb B_0)$ $\forall j$ in operator notation as
\begin{align*}
\D{0}( \mathcal U  (P_j^{(0)} )) \succ \D{0}( \mathcal V  (P_j^{(0)} )) \quad  \forall j
\end{align*}
from which convertibility follows.
\\ \\ 
\textbf{(ii)} Let us begin with the relative entropy of coherence. We have
\begin{align*}
\cre{1} = S(\rho_0 \, \| \, \D{1} \rho_0) &= - S(\rho_0) - \Tr \left( \rho_0 \log [ \D{1}( \rho_0)] \right) \\ 
&= - S(\rho_0) - \Tr \left( \D{1}(\rho_0) \log [ \D{1} (\rho_0)] \right) \\ 
&= S(\D{1}(\rho_0)) - S(\rho_0) \,\;.
\end{align*}
Since von Neumann entropy is a Schur-concave function, the assumption $\maj{1}{2}$ implies $S(\D{1}(\rho_0)) \le S(\D{2}(\rho_0))$ from which the claim follows.
\\ \\
In the following, we use the operator 2-norm $\left\|X \right\|_2 \coloneqq \sqrt{\Tr \left( X^\dagger X \right) }$. We have,
\begin{align*}
c^{(2)}_{\mathbb B_2} (\rho_0) &= \left\|  (\mathcal I - \D{2} ) \rho_0 \right\|^2_2 = \left\| \rho_0 \right\|_2^2 - \left\| \D{2} \rho_0 \right\|^2_2 = \left\| \rho_0 \right\|_2^2 - \left\|\mathcal U \left(\prod_\alpha \D{\alpha} \right) \D{1} \rho_0 \right\|^2_2  \\
&\ge \left\| \rho_0 \right\|_2^2 - \left\| \D{1} \rho_0 \right\|^2_2 = c^{(2)}_{\mathbb B_1} (\rho_0)  \,\;.
\end{align*}
The inequality follows since the 2-norm is submultiplicative over unital CPTP maps.
\end{proof}

\begin{proposition*}
$Q_{\mathbb B_0}(\mathbb B_1) \le q(\mathbb B_1,\mathbb  B_0)$.
\end{proposition*}
\begin{proof}
One has for $A =\sum_k a_k P_k^{(1)}$,
\begin{align*}
\Var_i (A) &=  \Tr \left( P_i^{(0)} A^2 \right) - \left[ \Tr \left( P_i^{(0)} A \right) \right]^2 \\
 & = \sum_k  a_k^2 \Tr \left(  P_i^{(0)} P_k^{(1)} \right) 
  - \sum_{k,l} a_k a_l \Tr \left(  P_i^{(0)} P_k^{(1)} \right) \Tr \left(  P_i^{(0)} P_l^{(1)} \right) \,\;,
\end{align*}
hence
\begin{align*}
Q_{\mathbb B_0}(\mathbb B_1)& \le \sup_{A \in \mathcal A_{\mathbb B_1}, \left\|A  \right\|_2 = 1} \sum_i \Var_i (A) \\
& =  \sup_{A \in \mathcal A_{\mathbb B_1}, \left\|A  \right\|_2 = 1} \left( 1 - \left\| X^T(\mathbb B_1,\mathbb B_0) \,\bs a  \right\|^2 \right) \\
& \le 1 - \lambda_{\min} \left( X(\mathbb B_1,\mathbb B_0)X^T(\mathbb B_1,\mathbb B_0) \right)
\end{align*}
which is the desired bound.
\end{proof}

\bounds*

\begin{proof}
We begin with the first inequality. If $\maj{1}{2}$, then there exists a bistochastic matrix $M$ such that $X(\mathbb B_2,\mathbb B_0) = M X(\mathbb B_1,\mathbb B_0)$. We need to show that this implies $\lambda_{\min} \left( X(\mathbb B_1,\mathbb B_0)X^T(\mathbb B_1,\mathbb B_0) \right) \coloneqq s_d^2 \left( X(\mathbb B_1,\mathbb B_0)\right) $ ($s_d$ denotes the minimum singular value) satisfies $s_d\left( X(\mathbb B_1,\mathbb B_0)\right) \ge s_d\left( X(\mathbb B_2,\mathbb B_0)\right)$. Indeed, this is guaranteed by the Gel'fand-Naimark inequality which states that (for the singular values sorted in decreasing order) $\prod_{j=1}^k s_{i_j} (A B) \le \prod_{j=1}^k s_{j} (A) \prod_{j=1}^k s_{i_j} (B)$ for all $1 \le i_1 \le \dotso \le i_k \le n$ and $k=1,\dots,n$ (in our case we set $k=1$ and $i_1 = n$) \cite{bhatia2013matrix}. Notice that $s_1(M) = 1$ since $M$ is bistochastic.
\\ \\
For the second one, since $X(\mathbb B_2,\mathbb B_0) = M X(\mathbb B_1,\mathbb B_0)$ for bistochastic $M$, we have that $\max_{i,j} X_{ij} (\mathbb B_2,\mathbb B_0) \le \max_{i,j} X_{ij} (\mathbb B_1,\mathbb B_0)$. The result follows from the monotonicity of the $\log$ function. Symmetry follows from $X(\mathbb B_2,\mathbb B_1) = X^T (\mathbb B_1,\mathbb B_2)$.  
\end{proof}


\allbases*

\begin{proof}
We first rewrite the condition for $\majj{F}{G}$ in the following equivalent form.
\begin{subequations}
\begin{align}
X(\mathbb G , \mathbb B_0) &= M X(\mathbb F , \mathbb B_0) \quad \Longleftrightarrow \\
\Tr \left( G_i P^{(0)}_j \right) &= \Tr \left(\sum_k M_{ik} F_k P^{(0)}_j \right) \quad \forall i,j \quad \Longleftrightarrow  \label{Eq:app_second} \\
\D{0} \left( G_i  \right) &= \D{0} \left(\sum_k M_{ik} F_k \right) \quad \forall i \,\;. \label{Eq:app_third}
\end{align}
\end{subequations}
If $\mathbb F$ is a parent of $\mathbb G$ then Eq.~\eqref{Eq:app_third} holds for all $\mathbb B_0$ with $M$ that is independent of $\mathbb B_0$.
\\ \\
For the converse, let $\D{0} \left( G_i  \right) = \D{0} \left(\sum_k M_{ik} F_k \right)$ $\forall i $ and $\forall\, \mathbb B_0 \in \mathcal M (\mathcal H)$ with $M$ that is independent of $\mathbb B_0$. Since all for any two bases there is always a unitary superoperator connecting them, the $\mathbb B_0$ freedom amount to inserting an arbitrary unitary in Eq.\eqref{Eq:app_second} as
\begin{align*}
\Tr \left( G_i P^{(0)}_j \right)  &=  \Tr \left( \textstyle \sum_k \displaystyle M_{ik} F_k  P^{(0)}_j  \right) \quad \forall i,j  \,\;.
\end{align*}
Now we show that the above implies $G_i = \sum_k M_{ik} F_k$. Notice that both $G_i$ and $\sum_k M_{ik} F_k$ are non-negative operators, hence also Hermitian. The above equation forces the (Hermitian) $G_i$ and $\sum_k M_{ik} F_k$ to have the same expectation value over all pure states, hence $G_i = \sum_k M_{ik} F_k$ $\forall i$.

\end{proof}

\end{document}